\documentclass[11pt,letterpaper]{article}
\usepackage[margin=1in]{geometry}
\usepackage{amsmath,amssymb,amsthm}
\usepackage{hyperref}
\usepackage{enumitem}

\newtheorem{lemma}{Lemma}

\newtheorem{proposition}{Proposition}
\newtheorem{theorem}{Theorem}
\newtheorem{corollary}{Corollary}

\title{Diagonalizing Through the $\omega$-Chain:\\Iterated Self-Certification on Bounded Turing Machines\\and its Least Fixed Point}
\author{Miara Sung\thanks{Independent scholar. Email at jbaek080@berkeley.edu}}

\begin{document}
\maketitle

\begin{abstract}
Bounded self-certification in Turing machines fails because self-simulation necessarily incurs a strictly positive temporal overhead. We translate this operational constraint into a domain-theoretic framework, defining an operator that advances a finite halting observation from time bound $i$ to $i+1$. While no bounded machine can achieve a fixed point under this operator, the iterative process forms an ascending $\omega$-chain. The Scott limit of this chain resolves to the least fixed point of the operator, representing an unbounded computation that fully captures the machine's halting behavior. Our construction provides a novel perspective on the halting problem, framing the transition from finite observability to the least fixed point as the continuous deferral of the diagonal.
\end{abstract}

\section{Setup}

Let $A$ be a Turing machine operating under a time bound $T$. We consider the self-halting problem for $A$: can $A$, on input $\langle A \rangle$ (the encoding of $A$), decide that $A$ halts within $T$ steps?

A Turing machine $D$ with time bound $K \gg T$ can answer this trivially by simulation. But $D$ with time bound $T$ cannot, even if $D = A$. To decide whether $A$ halts on $\langle A \rangle$, $D$ must simulate its execution. In general, simulating $k$ steps of an arbitrary Turing machine takes at least $k$ steps. Furthermore, at least one extra step is needed to enter a halting state and output the verdict.

\begin{lemma}
Let $A$ be an arbitrary Turing machine and let $D$ be any Turing machine (it is possible that $D = A$ by Kleene's recursion theorem) that correctly decides whether $A$ halts within $T$ steps. Then $D$ requires at least $T+1$ steps.
\end{lemma}

\begin{proof}
Let $c_0,c_1,\ldots$ denote the configuration sequence of $A$ on input $\langle A\rangle$, where $c_{t+1}=\delta_A(c_t)$. The statement ``$A$ halts within $T$ steps'' holds iff some $c_t$ with $t\le T$ is a halting configuration.

In the worst case, $A$ halts exactly at step $T$. To correctly decide the statement in this case, a machine $D$ must determine the configuration $c_T$. The only way to obtain $c_T$ from $c_0$ is to follow the computation of $A$ through $T$ successive transitions; there are no structural facts about $A$ that can help us, since $A$ is assumed to be arbitrary. Thus $D$ must simulate at least $T$ steps of $A$ before it can know whether halting occurs at step $T$.

After determining this information, $D$ must perform at least one additional step to enter a halting state and output the verdict. Therefore any correct decider executes at least $T+1$ steps. In particular, by Kleene's recursion theorem, there exists an $A = f(\langle A \rangle)$. Since we have shown it for the general case, we have shown that for the specific case where $D$ equals this $A = f(\langle A \rangle)$, a $T$-bounded machine cannot decide its own $T$-step halting.
\end{proof}

The consequence is immediate: $A$ with time bound $T$ cannot decide the halting behavior of $A$ at time $T$. At best, it can decide its behavior up to time $T-1$, but the $T$-th step remains opaque. $A$ cannot decide its own halting. The next step cannot be known until it is computed. This is a resource-bounded diagonal argument: attempting to certify the $T$-th step within time $T$ forces a strictly larger bound.

\section{The Operator}
Let $D$ be the domain of monotone partial functions $p : \mathbb{N} \to \{0, 1, \bot\}$, satisfying $p(k) = 1 \implies p(k') = 1$ for all $k' \geq k$, ordered pointwise by extension: $p \sqsubseteq q$ if $q$ agrees with $p$ wherever $p$ is defined. We interpret $p(k) = 1$ as ``$M$ halts at or before step $k$,'' $p(k) = 0$ as ``$M$ has not halted by step $k$,'' and $p(k) = \bot$ as ``the behavior at step $k$ is not yet determined.'' We define $F : D \to D$ by induction on $k$. For $k = 0$, the value $F(p)(0)$ is independent of $p$ and records a direct observation of $M$:
\[
F(p)(0) = \begin{cases}
1 & \text{if } M \text{ halts at or before step } 0, \\
0 & \text{otherwise.}
\end{cases}
\]
For $k \geq 0$, the value $F(p)(k+1)$ extends $p$ by one step:
\[
F(p)(k+1) = \begin{cases}
\bot & \text{if } p(k) = \bot, \\
1   & \text{if } p(k) = 1, \\
1   & \text{if } p(k) = 0 \text{ and } M \text{ halts at exactly step } k+1, \\
0   & \text{if } p(k) = 0 \text{ and } M \text{ does not halt at step } k+1.
\end{cases}
\]
Note that $F(p)(k+1)$ depends on $p$ only through the single value $p(k)$; this locality is what makes $F$ Scott-continuous, as shown in Proposition~\ref{prop:continuity}.

Starting from the empty partial model $p_0 = \bot$, define $p_{i+1} = F(p_i)$. By induction, $p_i$ is defined on $\{0, 1, \ldots, i-1\}$ and undefined beyond, and the sequence forms an ascending $\omega$-chain in $D$:
\[
p_0 \sqsubseteq p_1 \sqsubseteq p_2 \sqsubseteq \cdots
\]
Each application of $F$ records one additional observation of $M$'s configuration, consistent with the $+1$ overhead established in Lemma~1.

\begin{proposition}[Monotonicity and Continuity]
\label{prop:continuity}
The operator $F$ is monotone: more time yields strictly more information, and no previously determined value is contradicted. It is also Scott-continuous: for any $k$, the value $p_\omega(k)$ is determined at the finite stage $p_k$, so the limit commutes with $F$.
\end{proposition}
\begin{proof}
We first show monotonicity. Let $p, q \in D$ such that $p \sqsubseteq q$. By definition of the information ordering, $q$ agrees with $p$ wherever $p$ is defined ($p(x) \neq \bot \implies p(x) = q(x)$). The operator $F$ preserves all defined values of its input and determines at most one additional step of the deterministic computation of $M$. Because the underlying machine $M$ is deterministic, the additional computation step derived from $p$ must be identical to the one derived from $q$. Furthermore, any extra information in $q$ is preserved in $F(q)$. Therefore, $F(p) \sqsubseteq F(q)$.

Next, we show Scott continuity. Let $C \subseteq D$ be a directed set. We must show that $F(\bigsqcup C) = \bigsqcup_{p \in C} F(p)$. 
Because $F$ is monotone, the inequality $\bigsqcup_{p \in C} F(p) \sqsubseteq F(\bigsqcup C)$ holds universally. We only need to show the reverse inclusion: $F(\bigsqcup C) \sqsubseteq \bigsqcup_{p \in C} F(p)$.

Suppose $F(\bigsqcup C)(k) = v$ for some $k \in \mathbb{N}$ and $v \in \{0, 1\}$. The computation of $F$ at step $k$ depends only on a finite prefix of the computation of $M$ (specifically, the behavior up to step $k-1$). Because $C$ is directed, any finite subset of information contained in the supremum $\bigsqcup C$ must already be fully present in some single element $p^* \in C$. 

Since $p^*$ contains the necessary finite prefix to determine step $k$, it follows that $F(p^*)(k) = v$. Since $\bigsqcup_{p \in C} F(p)$ is the least upper bound, $F(p) \sqsubseteq \bigsqcup_{p \in C} F(p)$. In particular, $F(p^*) \sqsubseteq \bigsqcup_{p \in C} F(p)$. This implies that the value of $F(p^*)$, where it is defined, must equal the value of $\bigsqcup_{p \in C} F(p)$. Therefore, $\left( \bigsqcup_{p \in C} F(p) \right)(k) = v$. Thus, $F(\bigsqcup C)$ contains no information not already present in $\bigsqcup_{p \in C} F(p)$, meaning $F(\bigsqcup C) \sqsubseteq \bigsqcup_{p \in C} F(p)$. The two sides are equal, so $F$ is Scott-continuous.
\end{proof}

\section{The Fixed Point}

By Kleene's fixed point theorem, since $F$ is Scott-continuous on $D$, the least fixed point is:
\[
p_\omega = \mathrm{lfp}(F) = \bigsqcup_{n < \omega} p_n = \bigsqcup_{n < \omega} F^n(p_0)
\]

This is the total function that resolves $M$'s halting behavior at every finite step: for each $k$, $p_\omega(k)$ is defined, because $p_{k+1}$ already determines it. The fixed point $p_\omega$ is the complete halting observation of $M$.

But $p_\omega$ is not finitely computable. Each $p_n$ is produced by a machine with time bound $n$---a finite computation. The limit $\bigsqcup_n p_n$ requires taking a supremum over infinitely many such computations. The Turing machine that computes $p_\omega$ has time bound $\omega$. It is no longer a bounded Turing machine but a Turing machine requiring unbounded time. The iterated self-simulation, carried to its fixed point, has produced a transition from finite to infinite computation.

We note two theorems related to this fixed point. Let $B_T$ denote the class of finite halting observations computable by a machine with time bound $T$; that is, partial functions defined on $\{0, 1, \ldots, T-1\}$ and undefined beyond.

\begin{theorem}[No bounded fixed point]
\label{thm:no-bounded-fp}
For every finite time bound $T$, no machine in the class $B_T$ is a fixed point of $F$.
Equivalently,
\[
\forall T \in \mathbb{N}\;\; \forall p \in B_T \;\; (F(p) \neq p).
\]
\end{theorem}
\begin{proof}
Suppose for contradiction that $p \in B_T$ is a fixed point of $F$, i.e.\ $F(p) = p$. By definition, $p$ is computed by a machine with time bound $T$, so $p$ is defined on $\{0, 1, \ldots, T-1\}$ and $p(k) = \bot$ for all $k \geq T$. But $F(p)$ extends $p$ by one step: $F(p)$ is defined on $\{0, 1, \ldots, T\}$, since the machine with time bound $T+1$ can simulate the $T$-bounded machine to completion and resolve one additional step. In particular, $F(p)(T) \neq \bot$ while $p(T) = \bot$, so $F(p) \neq p$. Contradiction.
\end{proof}

\begin{theorem}[Least fixed point and unboundedness]
\label{thm:lfp}
Let $(p_n)_{n\in\mathbb{N}}$ be the $\omega$-chain defined by $p_0=\bot$ and $p_{n+1}=F(p_n)$.
Then
\[
\operatorname{lfp}(F)=\bigsqcup_{n\in\mathbb{N}} p_n = p_\omega,
\]
and moreover
\[
p_\omega \notin \bigcup_{T\in\mathbb{N}} B_T.
\]
\end{theorem}
\begin{proof}
The first claim follows directly from Kleene's fixed point theorem: $F$ is Scott-continuous on $D$ (Proposition~\ref{prop:continuity}), so $\operatorname{lfp}(F) = \bigsqcup_{n \in \mathbb{N}} F^n(\bot) = \bigsqcup_n p_n = p_\omega$.

For the second claim, suppose for contradiction that $p_\omega \in B_T$ for some $T \in \mathbb{N}$. Then $p_\omega$ is defined on at most $\{0, \ldots, T-1\}$. But for every $k \in \mathbb{N}$, the stage $p_{k+1}$ defines $p_{k+1}(k)$, so $p_\omega(k) = p_{k+1}(k) \neq \bot$. In particular, $p_\omega(T) \neq \bot$, which contradicts $p_\omega \in B_T$. Since $T$ was arbitrary, $p_\omega \notin \bigcup_{T \in \mathbb{N}} B_T$.
\end{proof}

\begin{corollary}[Bounded self-certification vs.\ Scott limit]
\label{cor:bounded-vs-limit}
Bounded self-certification fails while the Scott limit succeeds: no bounded machine
realizes a fixed point of $F$, whereas the directed supremum $p_\omega=\operatorname{lfp}(F)$
does.
\end{corollary}
\begin{proof}
By Theorem~\ref{thm:no-bounded-fp}, no $p \in B_T$ satisfies $F(p) = p$, so bounded self-certification fails. By Theorem~\ref{thm:lfp}, $p_\omega = \operatorname{lfp}(F)$ exists and satisfies $F(p_\omega) = p_\omega$, but $p_\omega \notin \bigcup_T B_T$: the fixed point is realized only at the Scott limit, which is not a bounded computation.
\end{proof}

\begin{theorem}[Semi-decidability as Asymmetric Observability]
\label{thm:semi-decidability}
The global halting property $H(M) \equiv \exists k \in \mathbb{N} \; (p_\omega(k) = 1)$ is finitely observable if true, but requires strictly infinite computation if false. The attempt to finitely compute $H(M) = \text{false}$ reconstructs the classical Turing diagonal contradiction via the $+1$ simulation overhead.
\end{theorem}
\begin{proof}
We evaluate the global halting property by inspecting the least fixed point $p_\omega = \bigsqcup_n p_n$. There are two cases:
\begin{enumerate}
    \item \textbf{Positive Case (Halts):} Suppose $M$ halts at some finite exact step $K$. Then $p_{K+1}(K) = 1$. Because the stage $p_{K+1}$ is computed by $M^{K+1}$ with a finite time bound $K+1$, the affirmative answer is resolved in finite time. The property $\exists k \; (p_\omega(k) = 1)$ is witnessed by a finite prefix of the $\omega$-chain.
    \item \textbf{Negative Case (Does Not Halt):} Suppose $M$ never halts. Then for all finite $k$, $p_\omega(k) = 0$. To affirmatively declare $H(M) = \text{false}$, a decider must verify the entire infinite sequence of $0$s. Suppose for contradiction there exists a machine $D$ with a finite time bound $T$ that correctly decides $H(M) = \text{false}$. 
    
    If such a $D$ exists, it compresses the infinite evaluation of the $\omega$-chain into a finite bound $T$. We can then construct a classical Turing diagonalizer $X$ that simulates $D(\langle X \rangle)$ and halts if and only if $D$ outputs ``does not halt''. 
    
    By Lemma 1, $X$ simulating $D$ incurs a strictly positive operational overhead. To evaluate $D$'s output and invert it, $X$ requires at least $T+1$ steps. Because $D$ bounded its analysis to time $T$, it failed to account for the $(T+1)$-th step where $X$ definitively contradicts $D$'s prediction. 
\end{enumerate}

Thus, the finite $+1$ overhead established in Lemma 1 acts as a local obstruction at every step $i$. For terminating programs, the overhead is bounded (stopping at $K+1$), so taking the supremum requires only finite computation. For non-terminating programs, the $+1$ overhead applies infinitely. Attempting to truncate this process to definitively say ``does not halt'' at any finite time $T$ immediately invites diagonalization at step $T+1$. Therefore, the continuous deferral of the diagonal paradox is exactly what pushes the negative decision into the uncomputable infinite limit, recovering the semi-decidability of the halting problem.
\end{proof}
\section{Remarks}

\textbf{Lawvere's theorem.} The diagonal argument underlying the simulation overhead---$A$ cannot simulate itself within its own time bound---is an instance of Lawvere's fixed point theorem \cite{lawvere1969}. In a Cartesian closed category, if a point-surjective map $e : X \to Y^X$ exists, every endomorphism of $Y$ has a fixed point; contrapositively, if some endomorphism lacks a fixed point, no such surjection exists. The bounded machine's state space cannot surject onto the space of functions on itself. Our observation is that this obstruction drives an iterative process whose fixed point lies at a higher computational level.

\textbf{Generality.} The construction does not depend on the specific model of computation. Any bounded device---such as a circuit of fixed depth or a bounded register machine---faces the same overhead in self-simulation and admits the same iterative construction. The key ingredients are: (i) a finite resource bound, (ii) strictly positive simulation overhead, and (iii) domain-theoretic structure that lets the resulting chain converge.

\end{document}